\documentclass{AIMS}
\usepackage{latexsym,amsmath,amssymb}
\usepackage{supertabular}
\usepackage{longtable}
  \usepackage{paralist}
  \usepackage{graphics} %% add this and next lines if pictures should be in esp format
  \usepackage{epsfig} %For pictures: screened artwork should be set up with an 85 or 100 line screen
 \usepackage[colorlinks=true]{hyperref}
\hypersetup{urlcolor=blue, citecolor=red}

\def\currentvolume{X}
 \def\currentissue{X}
  \def\currentyear{200X}
   \def\currentmonth{XX}
    \def\ppages{X--XX}

\newtheorem{theorem}{Theorem}[section]

\newtheorem{lemma}[theorem]{Lemma}
\newtheorem{proposition}[theorem]{Proposition}

\newtheorem{example}[theorem]{Example}

\theoremstyle{definition}
\newtheorem{definition}[theorem]{Definition}
\newtheorem{remark}[theorem]{Remark}

\newcommand{\N}{\mathbb N}

\newcommand{\F}{\mathbb F}

\newcommand{\ovom}{{\ensuremath{\overline{\omega}}}}
\newcommand{\om}{{\ensuremath{\omega}}}
\def\wt{\mathop{{\rm wt}}}
\def\wth{\textnormal{{\rm wt}}}

\def\tr{\mathop{{\rm tr}}}
\def\H{\mathop{{\rm H}}}

\def\dist{\mathop{{\rm dist}}}

\def\min{\mathop{{\rm min}}}
\def\F{{\mathbb F}}
\def\u{{\mathbf{u}}}
\def\v{{\mathbf{v}}}

\def\0{{\mathbf{0}}}

\newcommand{\cR}{\mathcal{R}}

\numberwithin{equation}{section}
\title[From Skew-Cyclic Codes to Asymmetric Quantum Codes]
      {From Skew-Cyclic Codes to Asymmetric Quantum Codes}

\author[Martianus Frederic Ezerman, San Ling, Patrick Sol{\'e} and Olfa Yemen]{}


 \email{mart0005@ntu.edu.sg}
 \email{lingsan@ntu.edu.sg}
 \email{sole@enst.fr}
 \email{olfa\underline{ }yemen@yahoo.fr}
\thanks{}

\begin{document}
\maketitle

% Enter the first and second authors' names and address:
\centerline{\scshape Martianus Frederic Ezerman and San Ling}
\medskip
{\footnotesize
% please put the address of the first and second authors
 \centerline{Division of Mathematical Sciences,}
   \centerline{School of Physical and Mathematical Sciences, Nanyang Technological University,}
   \centerline{21 Nanyang Link, Singapore 637371, Republic of Singapore}
} % Do not forget to end the {\footnotesize by the sign }

\medskip

\centerline{\scshape Patrick Sol{\'e}}
\medskip
{\footnotesize
 % please put the address of the third author
 \centerline{Centre National de la Recherche Scientifique (CNRS/LTCI),}
   \centerline{Telecom-ParisTech, Dept Comelec,}
   \centerline{46 rue Barrault, 75634 Paris, cedex France}
}

\medskip

\centerline{\scshape Olfa Yemen}
\medskip
{\footnotesize
 % please put the address of the fourth author
 \centerline{Institut Pr{\'e}paratoire aux {\'E}tudes d'Ing{\'e}nieurs El Manar,}
   \centerline{Campus Universitaire El Manar, Tunis, Tunisia}
}

\bigskip

% The name of the associate editor will be entered by an editorial staff
 \centerline{(Communicated by the associate editor name)}

%The abstract of your paper
\begin{abstract}
%Warning: maximum 200 words
We introduce an additive but not $\F_{4}$-linear map $S$ from $\F_{4}^{n}$ 
to $\F_{4}^{2n}$ and exhibit some of its interesting 
structural properties. If $C$ is a linear $[n,k,d]_4$-code, then 
$S(C)$ is an additive $(2n,2^{2k},2d)_4$-code. If $C$ is an additive 
cyclic code then $S(C)$ is an additive quasi-cyclic code of index $2$.
Moreover, if $C$ is a module $\theta$-cyclic code, a recently 
introduced type of code which will be explained below, then $S(C)$ is 
equivalent to an additive cyclic code if $n$ is odd and 
to an additive quasi-cyclic code of index $2$ if $n$ is even.
Given any $(n,M,d)_4$-code $C$, the code $S(C)$ is 
self-orthogonal under the trace Hermitian inner product. Since 
the mapping $S$ preserves nestedness, it can be used as a tool 
in constructing additive asymmetric quantum codes.
\end{abstract}

%The title of your section 1
\section{Introduction}\label{sec:Intro}
The class of \textit{skew-cyclic codes} was introduced in~\cite{BGU07}. 
These linear codes have the property of being invariant under the 
operation of cyclic shift composed with overall conjugation. 
Demanding an ideal structure on the codes forces us, over $\F_{4}$, 
to work in even lengths only. By relaxing this structure to that of a 
module~\cite{BU09}, it is now possible to deal with skew-cyclic 
codes of any lengths.

In the present work, a mapping $S$ is introduced to map any skew-cyclic 
codes of length $n$ over $\F_{4}$ into codes of length $2n$ which are 
invariant under a coordinate permutation denoted by $\sigma$. 
The permutation $\sigma$ is a cyclic permutation for $n$ odd and a 
product of two cycles of equal length for $n$ even.

Besides these structural properties, the mapping $S$ has interesting 
duality properties and preserves nestedness. These allow us to construct 
\textit{asymmetric quantum codes} following the method given in~\cite{ELS10}.

The material is organized as follows. In Section~\ref{sec:Prelims}, 
we state some basic definitions and properties of linear and 
additive codes. More specifically, the two families, $\mathbf{4^{\H}}$ 
and $\mathbf{4^{\H+}}$, of codes over $\F_{4}$ are formally defined. 
Their respective dualities and weight enumerators are stated.

Section~\ref{sec:MapS} introduces the mapping $S$ and its basic properties. 
The definition of and some algebraic background on module $\theta$-cyclic 
codes are discussed in Section~\ref{sec:Skew}. The study of the images of 
these codes under the mapping $S$ is also given. A very brief 
introduction to asymmetric quantum codes follows in Section~\ref{sec:AQC}.

In Section~\ref{sec:AnalysisWE}, an analysis of the weight enumerators 
is performed. This is important in understanding the parameters of the 
asymmetric quantum codes that can be obtained under the mapping $S$. 
Two systematic constructions of asymmetric quantum codes are given in 
Section~\ref{sec:Construct}. The one based on best-known linear codes 
is presented in Subsection~\ref{subsec:BKLC} while the other one, based on 
concatenated Reed-Solomon codes, is given in 
Subsection~\ref{subsec:RSCodes}. The last section contains conclusions 
and open problems.

\section{Preliminaries}\label{sec:Prelims}
Let $p$ be a prime and $q = p^m$ for some positive integer $m$. 
An $[n,k,d]_q$-linear code $C$ of length $n$, dimension $k$, 
and minimum distance $d$ is a subspace of dimension $k$ of the 
vector space $\F_q^n$ over the finite field $\F_q=GF(q)$ with 
$q$ elements. For a general, not necessarily linear, code $C$, 
the notation $(n,M=|C|,d)_q$ is commonly used.

A linear $[n,k,d]_q$-code $C$ is said to be \textit{cyclic} if $C$ 
is invariant under the cyclic shift. That is, whenever 
$\v=(v_0,v_1,\ldots,v_{n-2},v_{n-1}) \in C$, we have
\begin{equation*}
\v'=(v_{n-1},v_0,v_1\cdots,v_{n-2}) \in C \text{.}
\end{equation*}

Let $n$ be a positive integer and let $1 \leq l < n$ be a divisor of $n$. 
A linear $[n,k,d]_q$-code $C$ is \textit{quasi-cyclic of index $l$} or 
\textit{$l$-quasi-cyclic} if
\begin{equation*}
\v''=(v_{n-l},v_{n-l+1},\ldots,v_{n-1},v_{0},v_{1},\ldots,v_{n-l-1}) \in C
\end{equation*}
whenever $\v=(v_0,v_1,\ldots,v_{n-1}) \in C$. In particular, a 
$1$-quasi-cyclic code is a cyclic code.

As is the case for linear codes, we define the notions of an 
\textit{additive cyclic code} and an \textit{additive quasi-cyclic code} 
similarly by requiring the code to be additive, instead of linear.

The \textit{Hamming weight} of a vector or a codeword $\v$ in a code 
$C$, denoted by $\wt_H(\v)$, is the number of its nonzero entries. 
Given two elements $\u,\v \in C$, the number of positions where 
their respective entries disagree, written as $\dist_H(\u,\v)$, 
is called the \textit{Hamming distance} of $\u$ and $\v$. For any 
code $C$, the \textit{minimum distance} $d = d(C)$ is given by 
$d = d(C) = \min\left\lbrace \dist_H(\u,\v): 
\u,\v \in C,\u \neq \v\right\rbrace$. If $C$ is additive, then its 
additive closure property implies that $d(C)$ is given by the 
minimum Hamming weight of the nonzero vectors in $C$.

\begin{definition}\label{def1.1}
Let $\F_4:=\left\lbrace 0,1,\omega,\omega^{2}=\overline{\omega}\right\rbrace $. 
For $x \in \F_{4}$, set $\overline{x}=x^{2}$, the conjugate of $x$. 
Let $n$ be a positive integer and 
$\u = (u_0,u_1,\ldots,u_{n-1}), \v = (v_{0},v_{1},\ldots,v_{n-1}) \in \F_{4}^n$.
\begin{enumerate}
\item $4^{\H}$ is the family of $\F_{4}$-linear codes of 
length $n$ equipped with the \textit{Hermitian inner product} 
\begin{equation}\label{eq:1.1}
 \left\langle \u,\v\right\rangle _{\H} := 
\sum_{i=0}^{n-1} u_i \cdot v_i^{2} \text{.}
\end{equation}

\item $4^{\H+}$ is the family of $\F_{2}$-linear codes 
over $\F_{4}$ of length $n$ equipped with the 
\textit{trace Hermitian inner product} 
\begin{equation}\label{eq:1.2}
 \left\langle \u,\v\right\rangle _{\tr} := 
\sum_{i=0}^{n-1} (u_i\cdot v_i^{2} + u_i^{2} 
\cdot v_i) \text{.}
\end{equation}
\end{enumerate}
\end{definition}

\begin{definition}\label{def:1.1a}
A code $C$ of length $n$ over $\F_{4}$ is said to be an additive 
$\F_{4}$-code if $C$ belongs to the family $4^{\H+}$.
\end{definition}

Let $C$ be a code. Under a chosen inner product $*$, the 
\textit{dual code} $C^{\perp_{*}}$ of $C$ is given by
\begin{equation*}
 C^{\perp_{*}} := \left\lbrace \u \in \F_q^n : \left\langle 
\u,\v\right\rangle _{*} = 0 \text{ for all } \v \in C 
\right\rbrace \text{.}
\end{equation*}

A code is said to be \textit{self-orthogonal} 
if it is contained in its dual and is said to be \textit{self-dual} 
if its dual is itself. We say that a family of codes is 
\textit{closed} if $(C^{\perp_{*}})^{\perp_{*}} = C$ for each $C$ 
in that family. It has been established~\cite[Ch. 3]{NRS06} that both 
families of codes in Definition~\ref{def1.1} are closed.

The weight distribution of a code and that of 
its dual are important in the studies of their properties.

\begin{definition}\label{def1.2}
The \textit{weight enumerator} $W_C(X,Y)$ of an $(n,M=|C|,d)_q$-code 
$C$ is the polynomial
\begin{equation}\label{WE}
W_C(X,Y)=\sum_{i=0}^n A_{i} X^{n-i}Y^{i} \text{,}
\end{equation}
where $A_{i}$ is the number of codewords of weight $i$ in the code $C$.
\end{definition}

The weight enumerator of the Hermitian dual code $C^{\perp_{\H}}$ of an 
$[n,k,d]_4$-code $C$ is connected to the weight enumerator 
of the code $C$ via the \textit{MacWilliams Equation}
\begin{equation}\label{eq:MacW}
 W_{C^{\perp_{\H}}} (X,Y)= \frac{1}{|C|} W_C(X+3Y,X-Y) \text{.}
\end{equation}

From~\cite[Sec. 2.3]{NRS06} we know that the family $4^{\H+}$ has 
the same MacWilliams Equation as does the family $4^{\H}$. 
Thus, 
\begin{equation}\label{eq:MacWadd}
 W_{C^{\perp_{\tr}}} (X,Y)= \frac{1}{|C|} W_C(X+3Y,X-Y) \text{.}
\end{equation}

\section{The Mapping $S$ on Codes over $\F_{4}$}\label{sec:MapS}
Codes belonging to the family $4^{\H+}$ have been studied primarily 
in connection to designs (e.g.~\cite{KP03}) and to stabilizer 
quantum codes (e.g.~\cite[Sec. 9.10]{HP03}). 
It is well known that if $C$ is an additive $(n,2^k)_4$-code, 
then $C^{\perp_{\tr}}$ is an additive $(n,2^{2n-k})_4$-code.

Note that if the code $C$ is $\F_4$-linear with parameters $[n,k,d]_4$, 
then $C^{\perp_{\H}} = C^{\perp_{\tr}}$. This is because 
$C^{\perp_{\H}} \subseteq C^{\perp_{\tr}}$ and $C^{\perp_{\H}}$ is 
of size $4^{n-k} = 2^{2n-2k}$ which is also the size of $C^{\perp_{\tr}}$.

We are now ready to introduce the mapping $S$ in aid of 
later constructions.
\begin{definition}\label{def:4.1}
In $\F_4^{n}$, define the mapping
\begin{equation}\label{eq:4.2}
\begin{aligned}
S :\quad &\F_4^n\rightarrow\F_4^{2n}\\
&(v_{0},v_{1},\ldots,v_{n-1})\mapsto (v_0,\overline{v_0},v_1,
\overline{v_1},\ldots,v_{n-1},\overline{v_{n-1}}) \text{.}
\end{aligned}
\end{equation}
\end{definition}

It is immediately clear from the definition that $S$ is an 
$\mathbb{F}_{2}$-linear map, injective but not surjective.

\begin{example}\label{example:3.2}
The mapping $S$ is not $\F_4$-linear. Consider $n=2$ and 
$\mathbf{u}=(\omega,\overline{\omega})$. We have
\begin{equation*}
\begin{aligned}
S(\mathbf{u}) &= (\omega,\overline{\omega},\overline{\omega},\omega) \text{,}\\
S(w\cdot\mathbf{u}) &= S((\overline{\omega},1))=(\overline{\omega},\omega,1,1)\\
&\neq \omega\cdot S(\mathbf{u})=(\overline{\omega},1,1,\overline{\omega})\text{.}
\end{aligned}
\end{equation*}
\end{example}

\begin{lemma}\label{lemma:3.3}
Let $C$ be an $(n,M,d)_4$-code. 
For all $\u\in C$ we have
\begin{align*}
\wth_{H}(S(\u)) &= 2\wth_{H}(\u) \quad\text{and}\\
d(S(C)) &= 2d(C) \text{.}
\end{align*}
\end{lemma}

\begin{proof}
For all $u \in \F_4$, $S(u)=(u,\overline{u})$. Now, 
$\overline{u}=0$ if and only if $u=0$.
\end{proof}

The mapping $S$ is therefore a scaled isometry 
for the Hamming metric that preserves the code size. It sends an
additive $(n,M,d)_4$-code $C$ to an additive code $S(C)$ with 
parameters $(2n,M,2d)_{4}$.

\begin{lemma}
If $C$ is an additive $(n,M,d)_4$-cyclic code then $S(C)$ is an 
additive $(2n,M,2d)_4$-2-quasi-cyclic code.
\end{lemma}

\begin{proof}
Since $C$ is cyclic, 
\begin{equation*}
\v=(v_{0},\ldots,v_{n-2},v_{n-1}) \in C \textit{ if and 
only if } \v'=(v_{n-1},v_{0},\ldots,v_{n-2}) \in C \textit{.}
\end{equation*}
Applying $S$ yields
\begin{equation*}
\begin{aligned}
S(\v)=(v_{0},\overline{v_{0}},\ldots,v_{n-2},\overline{v_{n-2}},v_{n-1},\overline{v_{n-1}}) \in S(C)\text{,}\\
S(\v')=(v_{n-1},\overline{v_{n-1}},v_{0},\overline{v_{0}},\ldots,v_{n-2},\overline{v_{n-2}}) \in S(C)\text{.}
\end{aligned}
\end{equation*}
By definition, $S(C)$ is an additive $2$-quasi-cyclic code.
\end{proof}

\begin{proposition}\label{prop:1}
Given an additive $(n,M,d)_4$-code $C$, $S(C) \subseteq S(C)^{\perp_{\tr}}$.
\end{proposition}
\begin{proof}
Let $\v=(v_0,v_1,\ldots,v_{n-1}), \u=(u_0,u_1,\ldots,u_{n-1}) \in C$. 
Then
\begin{equation*}
\left\langle S(\v),S(\u) \right\rangle _{\tr} =
\sum_{i=0}^{n-1} \left( v_i \overline{u_i} + \overline{v_i} u_i\right) + 
\sum_{i=0}^{n-1} \left(\overline{v_i} u_i + v_i \overline{u_i} \right) =
2 \sum_{i=0}^{n-1}\left( v_i \overline{u_i} + \overline{v_i} u_i\right) = 0\text{.} 
\end{equation*}
\end{proof}

\section{Module $\theta$-Cyclic Codes over $\F_{4}$}\label{sec:Skew}
The motivation for our definition of \textit{module }$\theta$\textit{-cyclic codes} 
comes from~\cite{BGU07} and~\cite{BU09}. Given $\F_{q}$ and an 
automorphism $\theta$ of $\F_{q}$, we can define a ring structure 
on the set
\begin{equation*}
\cR=\F_{q}[X,\theta]=\left\lbrace a_{n}X^{n}+\ldots+a_{1}X+a_{0}|
a_{i}\in \F_{q} \text{ and } n\in \N \right\rbrace \text{.}
\end{equation*}
In $\cR$, the addition operation is the usual polynomial addition and 
the multiplication is defined by the extension to all elements of $\cR$, 
by associativity and distributivity, the basic rule 
$Xa = \theta(a)X$ for all $a \in \F_{q}$.

The ring $\cR$ is a left and right Euclidean ring whose left and right 
ideals are principal. Right division means that for nonzero 
$f(X),g(X) \in \cR$, there exist unique polynomials $Q_{r}(X),R_{r}(X) 
\in \cR$ such that
\begin{equation*}
 f(X)=Q_{r}(X)\cdot g(X) + R_{r}(X)
\end{equation*}
with $deg(R_{r}(X)) < deg (g(X))$ or $R_{r}(X)=0$.
If $R_{r}(X)=0$, then $g(X)$ is a \textit{right divisor} of $f(X)$ in $\cR$.

\begin{definition}\cite[cf. Defs. 1 and 3]{BU09}
Let $\theta$ be an automorphism of $\F_{q}$. Let $f(X) \in \cR$ be of 
degree $n$. If $I=(f(X))$ is a two-sided ideal of $\cR$, then an 
\textit{ideal $\theta$-code} $C$ is a left ideal $\cR g(X)/\cR f(X) 
\subset \cR/\cR f(X)$ where $g(X)$ is a right divisor 
of $f(X)$ in $\cR$. If $f(X)=X^{n}-1$, we call the ideal $\theta$-code 
corresponding to the left ideal $\cR g(X)/\cR (X^{n}-1) \subset \cR/\cR (X^{n}-1)$ 
an \textit{ideal $\theta$-cyclic code}. 

A \textit{module $\theta$-code} $C$ is a left 
$\cR$-submodule $\cR g(X)/\cR f(X) \subset \cR/\cR f(X)$ where 
$g(X)$ is a right divisor of $f(X)$ in $\cR$. Furthermore,
\begin{enumerate}
\item if $f(X)=X^{n}-c$, with $c \in \F_{q}$, we call the module 
$\theta$-code corresponding to the left $\cR$-module 
$\cR g(X)/\cR f(X) \subset \cR/\cR f(X)$ a 
\textit{module $\theta$-constacyclic code};
\item if $f(X)=X^{n}-1$, we call the module 
$\theta$-code corresponding to the left $\cR$-module 
$\cR g(X)/\cR f(X) \subset \cR/\cR f(X)$ a 
\textit{module $\theta$-cyclic code}.
\end{enumerate}
The length of the module $\theta$-code $C$ is $n=deg(f(X))$ and its dimension is 
$k=deg(f(X))-deg(g(X))$. If the minimum distance of $C$ is $d$, the 
code $C$ is said to be of type $[n,k,d]_{q}$.
\end{definition}
If the codewords of $C$ are identified with the list of the 
coefficients of the remainder of a right division by $f(X)$ in $\cR$, 
then the elements of $\cR g(X)/\cR f(X)$ are all of the left multiples 
of $g(X)=g_{r} X^{r}+\ldots+g_{1} X +g_{0}$.

Thus, a generator matrix $G$ of the corresponding module $\theta$-code of 
length $n=deg(f(X))$ is given by
\begin{equation}\label{eq:GenMatrix}
G=\left(
\begin{array}{*{12}{l}}
g_{0}  & g_{1}         & \ldots        & g_{r-1} & g_{r}           & 0             & \ldots & 0 \\
0      & \theta(g_{0}) & \theta(g_{1}) & \ldots  & \theta(g_{r-1}) & \theta(g_{r}) & \ldots & 0 \\
\vdots & \vdots        & \ddots        & \ddots  & \ddots          & \ddots        & \ddots & \vdots \\
0      & 0 & \ldots & 0  & \theta^{n-r-1}(g_{0}) & \ldots & \theta^{n-r-1}(g_{r-1}) & \theta^{n-r-1}(g_{r}) \\
\end{array}
\right)
\end{equation}
depending only on $g(X)$ and $n$.
\begin{theorem}
A module $\theta$-cyclic code $C_{\theta}$ has the following property
\begin{equation}
(v_{0},v_{1},\ldots,v_{n-2},v_{n-1}) \in C_{\theta} \Rightarrow 
(\theta(v_{n-1}),\theta(v_{0}),\theta(v_{1}),\ldots,\theta(v_{n-2})) \in C_{\theta} \textit{.}
\end{equation} 
\end{theorem}
\begin{proof}
The proof of this property for an ideal $\theta$-cyclic code $C$ is established 
in~\cite[Theorem 1]{BGU07}. The same proof works when we 
replace ideal by module.
\end{proof}

Since a module $\theta$-cyclic code $C_{\theta}$ has a representation in the skew 
polynomial ring $\cR=\F_{q}[X,\theta]$ (see~\cite{BU09}), when $\theta$ is fixed, 
we call $C_{\theta}$ a \textit{skew-cyclic code}.

We consider, for the rest of the paper, the Frobenius automorphism 
defined in $\F_4$ by $\theta(x)=x^2=\overline{x}$ for $x \in \F_{4}$. 
Let $[2n]$ denote the set $\left\lbrace 1,2,\ldots,2n\right\rbrace $. Let 
$\sigma = \tau \circ T^2$ be a permutation on $[2n]$ where $T$ is the 
cyclic shift module $2n$ and $\tau = (12)(34)\ldots(2n-1,2n)$. 
Since $T^2$ and $\tau$ commute, $\sigma$ can be written as $T^2 \circ \tau$ 
as well. We denote the identity permutation by $(1)$.

Let $\Sigma$ be the permutation on elements of $\F_4^{2n}$ induced by 
$\sigma$. That is, for $\v=(v_1,v_2,\ldots,v_{2n}) \in \F_4^{2n}$,
\begin{equation}\label{eq:induce}
\Sigma(\v)=\left(v_{\sigma(1)},v_{\sigma(2)},\ldots,v_{\sigma(2n)}\right)\text{.}
\end{equation}

\begin{lemma}\label{lemma:4.3}
Given an $(n,M,d)_{4}$-skew-cyclic code $C_{\theta}$, the 
code $S(C_{\theta})$ is invariant under $\Sigma$.
\end{lemma}
\begin{proof}
Let $\v=(v_1,v_2,\ldots,v_{2n}) \in S(C_{\theta})$. That is, there exists 
$\u=(u_{1},u_{2},\ldots,u_{n})\in C_{\theta}$ such that
\begin{equation*}
\v= (u_1,\overline{u_1},u_2,\overline{u_2},\ldots,u_{n},\overline{u_{n}})= S(\u)\text{.}
\end{equation*}
Since $C_{\theta}$ is a skew-cyclic code, we have 
\begin{equation*}
\overline{\u}:=(\overline{u_n},\overline{u_1},\ldots,\overline{u_{n-1}}) \in C_{\theta} \text{.}
\end{equation*}
Hence,
\begin{align*}
\Sigma(\v) & = (\overline{u_n}, u_n,\overline{u_1},u_1,\ldots,\overline{u_{n-1}},u_{n-1})\\
           & = S((\overline{u_n},\overline{u_1},\ldots,\overline{u_{n-1}}))\text{,}
\end{align*}
implying $\Sigma (S(C_{\theta})) \subset S(C_{\theta})$.
\end{proof}

\begin{lemma}\label{lemma:order}
The order of $\sigma$ is $2n$ if $n$ is odd and $n$ if $n$ is even.
\end{lemma}

\begin{proof}
For $1 \leq i \leq 2n$, $\sigma$ follows the following rule
\begin{equation}
\sigma : i \mapsto 
\begin{cases}
i + 3 \pmod{2n} \text{ if } i \text{ is odd} \\
i + 1 \pmod{2n} \text{ if } i \text{ is even}
\end{cases} \text{.}
\end{equation}
With computation done modulo $2n$, observe that if $i$ is odd, 
then $\sigma(i)=i+3$ and $\sigma^{2}(i)=\sigma(i+3)=i+4$. If 
$i$ is even, then $\sigma(i)=i+1$ and $\sigma^{2}(i)=\sigma(i+1)=i+4$.
Hence, $\sigma^2 =T^4 $.

Now, let $n=2l$ for some positive integer $l$. We have
\begin{equation*} 
\sigma^{n} =\sigma^{2l} =T^{4l}=T^{2n}=(1)\text{,}
\end{equation*}
the identity permutation. For $1 \leq k <n $, if $k=2i$, then 
\begin{equation*}
\sigma^k =\sigma^{2i} =T^{4i} \neq (1)
\end{equation*}
since $4i=2k<2n$. If $k=2i+1$, 
\begin{equation*}
\sigma^k =\sigma^{2i+1} =T^{4i} \circ \tau \neq(1)
\end{equation*}
since $\sigma^{k}$ sends $1$ to $4i+1 \neq 1$.
Consequently, the order of $\sigma$ is $n$.

In the case where $n=2l+1$, we have 
\begin{equation*}
\sigma^{2n} = (T^{2n})^2=(1)\text{.}
\end{equation*}
To show minimality, we first note that $\sigma^{n}=\tau$ 
since
\begin{equation*}
\sigma^{n}=\sigma^{2l+1}=T^{4l} \circ \sigma = T^{2n-2} \circ (\tau \circ T^{2}) 
= T^{2n-2} \circ (T^{2} \circ \tau)= \tau \text{.}
\end{equation*}
Consider the following two subcases. 
For $1 \leq k < n$, the same argument as in the even case above shows 
that $\sigma^k \neq (1)$. For $n+1 \leq k < 2n$,
\begin{equation*}
\sigma^{k}=\sigma^{n} \circ \sigma^{k-n} = \tau \circ \sigma^{k-n} \neq (1) \text{.}
\end{equation*}
We conclude that the order of $\sigma$ is $2n$.
\end{proof}

For conciseness, we adopt the following expressions following Lemma~\ref{lemma:order}.
\begin{enumerate}
\item For $n$ odd, $\sigma$ is the following cycle of length $2n$
\begin{equation}\label{eq:odd}
\sigma = (1,\sigma(1),\sigma^2(1),\ldots,\sigma^{2n-2}(1),\sigma^{2n-1}(1)) \text{.}
\end{equation}
\item Since $\sigma^k(1)\neq 2$ for all $0 \leq k \leq n-1$, for $n$ even, 
$\sigma$ can be written as the following product of two cycles, each of length $n$
\begin{equation}\label{eq:even}
\sigma = (1, \sigma(1),\sigma^2(1),\ldots,\sigma^{n-1}(1))
( 2,\sigma(2),\sigma^2(2),\ldots,\sigma^{n-1}(2)) \text{.}
\end{equation}
\end{enumerate}
When it is clear from the context, we write $C$ instead of $C_{\theta}$.

\begin{theorem}\label{thm:main_skew}
Let $C$ be an $[n,k,d]_4$-skew-cyclic code.
If $n$ is odd then $S(C)$ is equivalent to an additive $(2n,2^{2k},2d)_4$-cyclic code $C'$.
If $n$ is even then $S(C)$ is equivalent to an additive $(2n,2^{2k},2d)_4$-2-quasi-cyclic code $C'$.
\end{theorem}

\begin{proof}
Recall that the permutation $\sigma$ on $[2n]$ induces a permutation $\Sigma$ on 
the vectors of $\F_{4}^{2n}$. Consider first the case $n$ odd where 
Equation (\ref{eq:odd}) holds. 
Define the  permutation $\sigma'$ by
\begin{equation}\label{eq:3}
 \sigma' =
 \left(
    \begin{array}{ccccc}
      1 & 2 & \ldots & 2n-1 & 2n \\
      \sigma^{2n-1}(1) & \sigma^{2n-2}(1) & \ldots & \sigma(1) & 1 \\
    \end{array}
 \right).
\end{equation}
It is clear that for all $1 \leq j \leq 2n$,
\begin{equation}\label{eq:4}
 \sigma'(j)=\sigma^{2n-j}(1) \text{.}
\end{equation}
The permutation $\sigma'$ induces a permutation $\Sigma'$ acting 
on the elements of $\F_{4}^{2n}$. For $\v=(v_1,v_2,\ldots,v_{2n}) \in \F_{4}^{2n}$,
\begin{equation}\label{eq:induce2}
\Sigma'(\v)=\left(v_{\sigma'(1)},v_{\sigma'(2)},\ldots,v_{\sigma'(2n)}\right)\text{.}
\end{equation}
To show that $\Sigma'(S(C))$ is cyclic we must prove that for all codewords 
$\v \in S(C)$, $T(\Sigma'(\v))\in \Sigma'(S(C))$ where $T$ is the vector cyclic 
shift. Since $\Sigma(S(C))=S(C)$ by Lemma~\ref{lemma:4.3}, we only need to show that
\begin{equation}\label{eq:5}
 T(\Sigma'(\v))=\Sigma'(\Sigma(\v))\text{.}
\end{equation}
Let us start from the right hand side. By definition,
\begin{equation}\label{eq:6}
\Sigma(\v) = \left(v_{\sigma(1)},v_{\sigma(2)},\ldots,v_{\sigma (2n)}\right) 
:= (v'_1,v'_2,\ldots,v'_{2n})=\v'\text{.}
\end{equation}
From Equation (\ref{eq:4}), we know that
\begin{equation*}
\Sigma'(\Sigma(\v) ) = \left(v'_{\sigma'(1)},v'_{\sigma'(2)},\ldots,v'_{\sigma' (2n)}\right)
= \left( v'_{\sigma^{2n-1}(1)},v'_{\sigma^{2n-2}(1)},\ldots,v'_{\sigma^{1}(1)},v'_{\sigma^{0}(1)}\right) \text{.}
\end{equation*}
By Equation (\ref{eq:6}),
\begin{equation}\label{eq:8}
\Sigma'(\Sigma(\v))= \left(v_{1},v_{\sigma^{2n-1}(1)},\ldots,v_{\sigma^{2}(1)},v_{\sigma(1)}\right) \text{.}
\end{equation}
Moving on to the left hand side. Equation (\ref{eq:4}) implies 
\begin{equation*}
\Sigma'(\v) = \left(v_{\sigma'(1)},v_{\sigma'(2)},\ldots,v_{\sigma' (2n)}\right)
= \left(v_{\sigma^{2n-1}(1)},v_{\sigma^{2n-2}(1)},\ldots,v_{\sigma(1)},v_{1}\right) \text{.}
\end{equation*}
Applying the vector cyclic shift $T$ on $\Sigma'(\v)$ completes the proof of this case.

For $n$ even, Equation (\ref{eq:even}) holds. 
Let the permutation $\sigma''$ be given by
\begin{equation}\label{eq:9}
 \sigma'' =
 \left(
    \begin{array}{ccccccc}
      1 & 2 & 3 & 4 & \ldots & 2n-1 & 2n \\
      \sigma^{n-1}(1) & \sigma^{n-1}(2) & \sigma^{n-2}(1) & \sigma^{n-2}(2)& \ldots& \sigma^0(1) & \sigma^0(2) \\
    \end{array}
 \right).
\end{equation}
Let $b$ be an integer such that $1 \leq b \leq n$. For all $1 \leq j \leq 2n$,
\begin{equation}\label{eq:10}
\sigma''(j) =
\begin{cases}
\sigma^{n-b}(1) \text{ if } j=2b-1\\
\sigma^{n-b}(2) \text{ if } j=2b
\end{cases}\text{.}
\end{equation}

Let $\Sigma''$ be the permutation on vectors in $\F_{4}^{2n}$ induced by $\sigma''$. 
Applying $\Sigma''$ and by Equation (\ref{eq:6}), we have
\begin{align*}
\Sigma''(\Sigma(\v)) &= \left(v'_{\sigma''(1)},v'_{\sigma''(2)},v'_{\sigma''(3)},v'_{\sigma''(4)},
\ldots,v'_{\sigma'' (2n-1)},v'_{\sigma'' (2n)}\right)\\
&= \left(v'_{\sigma^{n-1}(1)},v'_{\sigma^{n-1}(2)},v'_{\sigma^{n-2}(1)},v'_{\sigma^{n-2}(2)},
\ldots,v'_{\sigma(1)},v'_{\sigma(2)} ,v'_{1},v'_{2}\right)
\end{align*}
by Equation (\ref{eq:10}). Now, Equation(\ref{eq:6}) allows us to write 
\begin{equation*}
\Sigma''(\Sigma(\v) )= \left(v_{1},v_{2},v_{\sigma^{n-1}(1)},v_{\sigma^{n-1}(2)},\ldots,v_{\sigma^{2}(1)},v_{\sigma^{2}(2)}
 ,v_{\sigma(1)},v_{\sigma(2)}\right) \text{.}
\end{equation*}
Since 
\begin{align*}
\Sigma''(\v) & =  \left(v_{\sigma''(1)},v_{\sigma''(2)},v_{\sigma''(3)},v_{\sigma''(4)},
\ldots,v_{\sigma'' (2n-3)},v_{\sigma'' (2n-2)},v_{\sigma'' (2n-1)},v_{\sigma'' (2n)}\right)\\
 &= \left(v_{\sigma^{n-1}(1)},v_{\sigma^{n-1}(2)},v_{\sigma^{n-2}(1)},v_{\sigma^{n-2}(2)},\ldots,v_{\sigma(1)},v_{\sigma(2)}
 ,v_{1},v_{2}\right)
\end{align*}
and $\Sigma(S(C))=S(C)$, we get
\begin{equation*}
T^2 (\Sigma''(\v)) = \Sigma''(\Sigma(\v)) \in \Sigma''(S(C))\text{.}
\end{equation*}
Thus, $\Sigma''(S(C))$ is a 2-quasi-cyclic code. This completes the entire proof.
\end{proof}

\begin{example}
For $n=4$, we have
\begin{align*}
\sigma &=(1,4,5,8)(2,3,6,7) \text{ and}\\
\sigma'' &=(1,8,2,7)(3,5,4,6) \text{.}
\end{align*}
Following~\cite[Example 2]{BGU07}, let $C$ be a $[4,2,3]_4$-skew-cyclic 
code with generator matrix
\begin{equation}\label{eq:n=4}
G=\left(
\begin{array}{*{12}{l}}
1  & 0  & \ovom  & \om      \\
0  & 1  & \om  & \ovom
\end{array}
\right)\text{.}
\end{equation}
Verifying that $S(C)$ is invariant under $\Sigma$ is immediate.

Choose $\u=(1,0,\ovom,\om) \in C$. Let $\v=S(\u)=(1,1,0,0,\ovom,\om,\om,\ovom)$. Then
\begin{align*}
\Sigma'' \left(\Sigma(\v)\right) &= \left(v_{\sigma^4(1)},v_{\sigma^4(2)},
v_{\sigma^3(1)},v_{\sigma^3(2)},v_{\sigma^2(1)},v_{\sigma^2(2)},v_{\sigma(1)},v_{\sigma(2)} \right)\\
                                 &= \left(v_{1},v_{2},v_{8},v_{7},v_{5},v_{6},v_{4},v_{3} \right)
= (1,1,\ovom,\om,\ovom,\om,0,0) \text{, while} \\
\Sigma''(\v) &= \left(v_{\sigma^3(1)},v_{\sigma^3(2)},
v_{\sigma^2(1)},v_{\sigma^2(2)},v_{\sigma(1)},v_{\sigma(2)},v_1,v_2 \right)\\
                                 &= \left(v_{8},v_{7},v_{5},v_{6},v_{4},v_{3},v_{1},v_{2} \right)
= (\ovom,\om,\ovom,\om,0,0,1,1) \text{.}
\end{align*}
\end{example}

Explicit computation up to length $n=21$ shows that the only examples of 
module $\theta$-cyclic codes of odd lengths are the usual cyclic codes.

\begin{example}
For $n=7$, we have
\begin{align*}
\sigma &=(1,4,5,8,9,12,13,2,3,6,7,10,11,14) \text{ and}\\
\sigma' &=(1,14)(2,11,8,13,4,7)(3,10,9,12,5,6) \text{.}
\end{align*}
Let $C$ be a $[7,4,3]_4$-skew-cyclic 
code with generator matrix
\begin{equation}\label{eq:n=7}
G=\left(
\begin{array}{*{12}{l}}
1    & 1   & 0  & 1  & 0 & 0 & 0 \\
0    & 1   & 1  & 0  & 1 & 0 & 0 \\
0    & 0   & 1  & 1  & 0 & 1 & 0 \\
0    & 0   & 0  & 1  & 1 & 0 & 1 
\end{array}
\right)\text{.}
\end{equation}

Let $\v=(1,1,1,1,0,0,1,1,0,0,0,0,0,0)$. Then
\begin{align*}
\Sigma' \left(\Sigma(\v)\right) &= \left(v_{1},v_{14},v_{11},v_{10},v_{7},v_{6},v_{3},v_{2},v_{13},v_{12},v_{9},v_{8},v_{5},v_{4} \right) \\
                                &= (1,0,0,0,1,0,1,1,0,0,0,1,0,1)\text{, while} \\
\Sigma'(\v) &= \left(v_{\sigma^{13}(1)},v_{\sigma^{12}(1)},v_{\sigma^{11}(1)},v_{\sigma^{10}(1)},\ldots, v_{\sigma(1)},v_{1}\right)\\
                                 &= \left(v_{14},v_{11},v_{10},v_{7},v_{6},v_{3},v_{2},v_{13},v_{12},v_{9},v_{8},v_{5},
v_{4},v_{1}\right)\\
                                 &= (0,0,0,1,0,1,1,0,0,0,1,0,1,1) \text{.}
\end{align*}
\end{example}

Theorem~\ref{thm:main_skew}, our main result in this section, reveals 
the structural connection between skew-cyclic codes under the mapping 
$S$ and additive cyclic or additive 2-quasi-cyclic codes, depending on 
the parity of the length. Combined with the orthogonality property that 
the mapping $S$ induces, we can further make a connection to 
asymmetric quantum codes.

\section{Asymmetric Quantum Codes}\label{sec:AQC}
For brevity, it is assumed that the reader is familiar 
with the standard error model in quantum error-correction, 
both symmetric and asymmetric. For references on the motivation 
and previous constructions of asymmetric quantum codes,~\cite{SRK09} 
and~\cite{WFLX09} can be consulted.

\begin{definition}
Let $d_{x}$ and $d_{z}$ be positive integers. A
quantum code $Q$ in $V_{n}=\mathbb{C}^{q^{n}}$ 
with dimension $K\geq2$ is called an 
\textit{asymmetric quantum code} with parameters $((n,K,d_{z}/d_{x}))_{q}$
or $[[n,k,d_{z}/d_{x}]]_{q}$, where $k=\log_{q}K$, if $Q$ detects $d_{x}-1$
quantum digits of $X$-errors and, at the same time, $d_{z}-1$
quantum digits of $Z$-errors.
\end{definition}

The following result has been shown recently in~\cite{ELS10}.
\begin{theorem}\cite[Th. 4.5]{ELS10}\label{thm:main}
Let $q=r^2$ be an even power of a prime $p$. For $i = 1,2$, 
let $C_i$ be a classical additive code with parameters 
$(n,K_i,d_i)_q$. If $C_1^{\perp_{\tr}} \subseteq C_2$, then 
there exists an asymmetric quantum code $Q$ with parameters 
$((n,\frac{|C_2|}{|C_1^{\perp_{\tr}}|},d_z/d_x))_q$ where 
$\left\lbrace d_z,d_x\right\rbrace = \left\lbrace d_1,d_2\right\rbrace$.
\end{theorem}

As explained in~\cite{BGU07} and in~\cite{BU09}, there are two major gains 
of using module $\theta$-codes. First, there is more flexibility and 
generality in constructing (linear) codes without increasing the complexity of 
the encoding and decoding process. The notion of \textit{$q$-cyclic codes}, 
introduced in~\cite{Ga85}, for instance, covers ideal $\theta$-cyclic 
codes with $\theta$ limited to the Frobenius automorphism only.

More important to the agenda of constructing asymmetric quantum codes 
is the second gain, which is the minimum distance improvement. Exhaustive 
search on module $\theta$-codes up to certain length has yielded 
linear codes with better parameters. More systematically, the BCH approach of 
constructing codes with a prescribed lower bound on the minimum distance can be 
extended to module $\theta$-codes as well. Section 3 of~\cite{BU09} contains 
the construction details. The resulting improvements have been added to the 
database of \textit{best-known linear codes} (BKLC) of MAGMA~\cite{BCP97}.

For the remaining of the paper, we will concentrate on constructing 
asymmetric quantum codes with $d_{z} \geq d_{x}=2$ based on Theorem~\ref{thm:main}. 
We will see how the mapping $S$ can be used as an aid in construction. 
All computations are done in MAGMA V2.16-5.

\section{Analysis on the Weight Enumerators}\label{sec:AnalysisWE}
In this section, the weight enumerators of $S(C)$ and of $S(C)^{\perp_{\tr}}$ 
are analyzed. This analysis will be useful in determining $d_{x}$.

Let $A_{i}$ be the number of codewords of weight $i$ in an 
additive $(n,M,d)_4$-code $C$. Then the weight enumerators of $S(C)$ 
and $S(C)^{\perp_{\tr}}$ can be written in terms of the weight 
enumerator of $C$ with the help of Equation (\ref{eq:MacWadd}) 

\begin{align}
 W_{S(C)}(X,Y) &=\sum_{i=0}^n A_{i} X^{2(n-i)}Y^{2i} \text{,} \label{eq:6.1} \\
 W_{S(C)^{\perp_{\tr}}}(X,Y) &= \frac{1}{|S(C)|} W_{S(C)}(X+3Y,X-Y) \text{.} \label{eq:6.2}
\end{align}

More explicitly,
\begin{equation}\label{eq:6.3}
W_{S(C)^{\perp_{\tr}}}(X,Y) = \frac{1}{M}\sum_{i=0}^{n} A_i L_{i} \text{,}
\end{equation}
where $L_{i}$ is given by
\begin{equation}\label{eq:6.4}
\left(\sum_{j=0}^{n-i}\binom{n-i}{j} X^{n-i-j}(3Y)^{j}\right)^{2}
\left(\sum_{l=0}^{i}\binom{i}{l} X^{i-l}(-Y)^{l}\right)^{2} \text{.}
\end{equation}

Denote the number of codewords of weight $i$ in the code 
$C^{\perp_{\tr}}$ by $A_i^{\perp_{\tr}}$. By using the 
\textit{Pless power moments} with $q=4$ (see~\cite[p. 259]{HP03} 
for the linear version), we have
\begin{equation}\label{eq:6.5}
\sum_{i=0}^n A_i =|C|=M \text{,}
\end{equation}

\begin{equation}\label{eq:6.6}
\sum_{i=0}^n iA_i =\frac{M}{4}(3n-A_1^{\perp_{\tr}}) \text{,}
\end{equation}

\begin{equation}\label{eq:6.7}
\sum_{i=0}^n i^{2}A_i =\frac{M}{4^{2}}\left\lbrace (9n^2+3n)-(6n-2)A_1^{\perp_{\tr}}
+2A_2^{\perp_{\tr}}\right\rbrace \text{.}
\end{equation}

If we further assume that $A_1^{\perp_{\tr}}=A_2^{\perp_{\tr}}=0$, 
then the following statements hold for Equation (\ref{eq:6.2}).
\begin{enumerate}
 \item The coefficient of $Y^{0}X^{2n}$ is $\frac{1}{M}\sum_{i=0}^{n}A_i = 1$.
 \item The coefficient of $YX^{2n-1}$ is
\begin{align*}
&\frac{1}{M} \sum_{i=0}^{n} A_i \left( 2 \cdot (n-i) \cdot 3 - 2i \right) \\
&= \frac{1}{M} \sum_{i=0}^{n} A_i(6n-8i) = 6n - 4^{-1}\cdot8(3n) = 0
\end{align*}
by Equation (\ref{eq:6.6}).
\item The coefficient of $Y^{2}X^{2n-2}$ is
\begin{align*}
&\frac{1}{M} \sum_{i=0}^{n} A_i \left(18n^2-48ni+32i^2-9n+8i\right) \\
&= \frac{18n^2-9n}{M}\sum_{i=0}^{n} A_i + \frac{8-48n}{M}
\sum_{i=0}^{n} iA_i + \frac{32}{M} \sum_{i=0}^{n}i^2A_i \\
&= 3n
\end{align*}
by Equations (\ref{eq:6.6}) and (\ref{eq:6.7}).
\end{enumerate}

If we rewrite
\begin{equation}\label{eq:6.8}
W_{S(C)^{\perp_{\tr}}}(X,Y) = \sum_{i=0}^{2n} B_i X^{2n-i}Y^i \text{,}
\end{equation}
then $B_0=1, B_1=0,$ and $B_2=3n$. This is true for any additive 
$(n,M,d)_4$-code $C$ with $d(C^{\perp_{\tr}})\geq 3$. 
If $d(C^{\perp_{\tr}})=1$, then $B_1=2A_1^{\perp_{\tr}}>0$. 
If $d(C^{\perp_{\tr}})=2$, then $B_1=0$ and $B_2 = 3n+4A_2^{\perp_{\tr}}>0$.

As a direct consequence of Proposition~\ref{prop:1} and the above 
analysis on the weight enumerators, we derive the following result.
\begin{proposition}\label{prop:2}
Given any additive $(n,M,d)_4$-code $C$ such that $d(C^{\perp_{\tr}}) \geq 2$, 
there exists an asymmetric quantum code $Q$ with parameters 
$[[2n,\log_{4}\left( \frac{|S(C)^{\perp_{\tr}}|}{|S(C)|}\right),2/2]]_4$.
\end{proposition}
\begin{proof}
By Proposition~\ref{prop:1}, $S(C) \subseteq S(C)^{\perp_{\tr}}$. 
Apply Theorem~\ref{thm:main} by taking $C_1 = C_2 = S(C)^{\perp_{\tr}}$. 
The values $d_z = d_x = 2$ follow from the analysis on the weight enumerators.
\end{proof}

The parameters of the resulting code $Q$ based on the construction in 
Proposition~\ref{prop:2} are not so good. Fortunately, the mapping $S$ 
preserves nestedness. This fact can be used to derive asymmetric quantum codes 
with better parameters.

\begin{theorem}\label{theorem:6.1}
Let $C$ be an additive $(n,M_1,d_1)_4$-code such that 
$d(C^{\perp_{\tr}})\geq 2$. Let $D$ be an additive 
$(n,M_2,d_2)_4$-code satisfying $C \subseteq D$. 
Then there exists an asymmetric quantum code $Q$ with 
parameters $[[2n,\log_{4}\left( \frac{M_{2}}{M_{1}}\right),2d_2/2]]_4$.
\end{theorem}

\begin{proof}
Apply Theorem~\ref{thm:main} by taking $C_1 = S(C)^{\perp_{\tr}}$ 
and $C_2 = S(D)$. The code $S(C)$ is an additive $(2n,M_{1},2d_1)_4$-code. 
Similarly, $S(D)$ is an additive code of parameters 
$(2n,M_{2},2d_2)_4$. The values for $d_z$ and $d_x$ follow 
from the discussion on the weight enumerators above.
\end{proof}

\begin{example}
Let $C=D$ be the $[n,1,n]_4$-repetition code generated 
by the all one vector $\mathbf{1}=(1,\ldots,1)$. 
It can be directly verified that $d(C^{\perp_{\tr}})= 2$. 
Hence, we get an asymmetric quantum code $Q$ with parameters 
$\mathbf{[[2n,0,2n/2]]_4}$ by Theorem~\ref{theorem:6.1}. This code 
$Q$ satisfies the equality of the quantum version of the 
Singleton bound $k \leq n-d_{x}-d_{z}+2$.
\end{example}

Henceforth, any asymmetric quantum code $Q$ satisfying 
$k = n-d_{x}-d_{z}+2$ is printed in boldface. 
We call such a code an \textit{asymmetric quantum MDS code}.

\begin{example}\label{example:5.2}
Consider the $[4,2,3]_4$-module $\theta$-cyclic code $D$ with generator 
matrix $G$ in Equation(\ref{eq:n=4}) above. The code $D$ contains the 
$[4,1,4]_4$-repetition code $C$ generated by $\mathbf{1}$. Applying 
Theorem~\ref{thm:main} with $C=C_{1}^{\perp_{\tr}}$ and $D=C_{2}$ 
results in a $\mathbf{[[4,1,3/2]]_{4}}$-asymmetric quantum code. 
Under the mapping $S$, by Theorem~\ref{theorem:6.1}, we arrive at 
an $[[8,1,6/2]]_{4}$-asymmetric quantum code.
\end{example}

The investigation on self-dual module $\theta$-code yields 
new Hermitian self-dual linear $\F_{4}$-codes with 
parameters $[50,25,14]_{4}$ and $[58,29,16]_{4}$. These codes 
are listed down in~\cite[Table 3]{BU09}. They can be used to derive 
asymmetric quantum codes $Q$ with parameters $[[50,0,14/14]]_{4}$ 
and $[[58,0,16/16]]_{4}$ following~\cite[Th. 7.1]{ELS10}. The latter 
code improves on the $[[58,0,14/14]]_{4}$-code in~\cite[Table III]{ELS10}.

The next section presents two systematic constructions of 
asymmetric quantum codes with $d_{z} \geq d_{x}=2$ by using the database of BKLC and by  
applying the mapping $S$ on concatenated Reed-Solomon codes, respectively.

\section{Two Constructions}\label{sec:Construct}
Under the mapping $S$, Theorem~\ref{theorem:6.1} says that 
while we cannot improve on $d_x=2$, we can relax the condition 
on the inner code $C$ to possibly improve on the size of $Q$ 
as well as on $d_{z}$. Our aim, then, is to choose the 
smallest possible subcode $C$ of $D$ such that 
$d(C^{\perp_{\tr}})\geq 2$ while keeping the size and 
the minimum distance of $D$ relatively large.

Note that there is no additive $(n,2,d)_{4}$-code with 
$d(C^{\perp_{\tr}})\geq 2$. The smallest additive code with 
$d(C^{\perp_{\tr}})= 2$ is an $(n,4,n)_{4} = [n,1,n]_{4}$-code $C$ 
consisting of the scalar multiples of a codeword $\v$ of weight 
$n$. Since this code $C$ is MDS, its dual 
$C^{\perp_{\tr}}=C^{\perp_{\H}}$ is of parameters $[n,n-1,2]_{4}$.

\subsection{Construction from best-known linear codes (BKLC)}\label{subsec:BKLC}
Let $n,k$ be fixed with $2 \leq k \leq n-1$. 
The strategy here is to consider the best-known 
linear code $D$ of length $n$ and dimension $k$ 
stored in the MAGMA database and check if the code contains 
codewords of weight $n$ and put them in a set $R$. If $R$ is non-empty, 
we choose an arbitrary codeword $\v \in R$ and construct a subcode 
$C \subset D$ of parameters $[n,1,n]_{4}$ whose elements are the 
four scalar multiples of $\v$.

Based on the codes $C$ and $D$, two asymmetric quantum codes can be derived, 
one from Theorem~\ref{thm:main} directly without the mapping $S$ by 
letting $C_{1}^{\perp_{\tr}}=C$ and $C_{2}=D$ and another from 
Theorem~\ref{theorem:6.1} under the mapping $S$. 
We label the first quantum code $Q$ while the second one $Q_{S}$.

\begin{theorem}\label{thm:MDSAsym}
Given any positive integer $n \geq 3$, there exists an 
$\mathbf{[[n,n-2,2/2]]_{4}}$-asymmetric quantum MDS code.
\end{theorem}
\begin{proof}
A general proof for the existence of an 
$\mathbf{[[n,n-2,2/2]]_{q}}$-asymmetric quantum MDS code is already 
given in~\cite[Cor. 3.4]{WFLX09}. Here we present a simple 
constructive proof for $q=4$. A cyclic code $D$ with parameters 
$[n,n-1,2]_4$ can be constructed by using $X+1$ as its generator 
polynomial. Its minimum distance is two since the check polynomial 
is $1+X+\ldots+X^{n-1}$. By~\cite[Th. 1]{EGS09}, $D$ has codewords 
of length $n$. One such codeword can be chosen to form an 
$[n,1,n]_{4}$-code $C$. Applying Theorem~\ref{thm:main} with 
$C_{1}^{\perp_{\tr}}=C$ and $C_{2}=D$ brings us to the conclusion.
\end{proof}

For a fixed $n$, it is not guaranteed that for all 
$k \in \left\lbrace 2,\ldots,n-2 \right\rbrace$, the best-known 
linear code with parameters $[n,k,d]_{4}$ has codewords of weight $n$. 
For example, there is no codeword of weight 6 in the best-known 
$[6,4,2]_{4}$-code stored in the database of MAGMA that we use here.

Table~\ref{table:BKLC3to20} lists down 
the resulting quantum codes for $n=4$ to $n=20$ based on the list of 
best-known linear codes with parameters $[n,k]_{4}$ invoked under 
the command BKLC in MAGMA. We exclude the case of $k=n-1$ 
in light of Theorem~\ref{thm:MDSAsym} and the case of $k=1$ due 
to~\cite[Ex. 8.2]{ELS10}. The process can of course be done for larger 
values of $n$ if so desired. Interested readers may contact the first 
author for the complete list of codes $Q$ and $Q_{S}$ with $d_{z} \geq 
d_{x}=2$ which are derived from the best-known linear codes for up to $n=46$.

\begin{table}[ht!]
\caption{Asymmetric QECC from BKLC}
\label{table:BKLC3to20}
\centering
\begin{tabular}{| c | l | l || c | l | l |}
\hline\hline
%\textbf{No.} & %
\textbf{$n$} & \textbf{Code $Q$} & \textbf{Code $Q_{S}$} &
%\textbf{No.} &%
\textbf{$n$} & \textbf{Code $Q$} & \textbf{Code $Q_{S}$} \\
\hline
$4$ & $\mathbf{[[4,1,3/2]]_{4}}$  & $[[8,1,6/2]]_{4}$	 & $15$ & $[[15,7,6/2]]_{4}$  &	$[[30,7,12/2]]_{4}$ \\
$5$ & $\mathbf{[[5,2,3/2]]_{4}}$  & $[[10,2,6/2]]_{4}$   &      & $[[15,8,5/2]]_{4}$  &	$[[30,8,10/2]]_{4}$ \\
$6$ & $\mathbf{[[6,2,4/2]]_{4}}$  & $[[12,2,8/2]]_{4}$   &      & $[[15,10,4/2]]_{4}$ &	$[[30,10,8/2]]_{4}$ \\
$7$ & $[[7,2,4/2]]_{4}$	      &	$[[14,2,8/2]]_{4}$ &    & $[[15,11,3/2]]_{4}$	      &	$[[30,11,6/2]]_{4}$ \\
    & $[[7,3,3/2]]_{4}$	      &	$[[14,3,6/2]]_{4}$ & $16$ & $[[16,2,12/2]]_{4}$	      &	$[[32,2,24/2]]_{4}$ \\

$8$ & $[[8,1,6/2]]_{4}$       & $[[16,1,12/2]]_{4}$ &    & $[[16,3,11/2]]_{4}$	      &	$[[32,3,22/2]]_{4}$ \\
  & $[[8,2,5/2]]_{4}$	      &	$[[16,2,10/2]]_{4}$ &    & $[[16,6,8/2]]_{4}$	      &	$[[32,6,16/2]]_{4}$ \\
  & $[[8,3,4/2]]_{4}$	      &	$[[16,3,8/2]]_{4}$ &    & $[[16,7,7/2]]_{4}$	      &	$[[32,7,14/2]]_{4}$ \\
  & $[[8,4,3/2]]_{4}$	      & $[[16,4,6/2]]_{4}$ &    & $[[16,8,6/2]]_{4}$	      &	$[[32,8,12/2]]_{4}$ \\
$9$ & $[[9,2,6/2]]_{4}$	      &	$[[18,2,12/2]]_{4}$ &    & $[[16,9,5/2]]_{4}$	      &	$[[32,9,10/2]]_{4}$ \\	
%with S is better than 5 Table VI of [ELS10]

  & $[[9,3,5/2]]_{4}$	      &	$[[18,3,10/2]]_{4}$ &    & $[[16,11,4/2]]_{4}$	      &	$[[32,11,8/2]]_{4}$ \\
  & $[[9,4,4/2]]_{4}$	      &	$[[18,4,8/2]]_{4}$  &    & $[[16,12,3/2]]_{4}$	      &	$[[32,12,6/2]]_{4}$ \\
  & $[[9,5,3/2]]_{4}$	      &	$[[18,5,6/2]]_{4}$  & $17$ & $[[17,5,9/2]]_{4}$	      & $[[34,5,18/2]]_{4}$ \\
$10$ & $[[10,3,6/2]]_{4}$     & $[[20,3,12/2]]_{4}$ &    & $[[17,8,7/2]]_{4}$	      & $[[34,8,14/2]]_{4}$ \\
   & $[[10,4,5/2]]_{4}$	      & $[[20,4,10/2]]_{4}$ &    & $[[17,9,6/2]]_{4}$	      & $[[34,9,12/2]]_{4}$ \\

   & $[[10,5,4/2]]_{4}$	      & $[[20,5,8/2]]_{4}$  &    & $[[17,10,5/2]]_{4}$	      & $[[34,10,10/2]]_{4}$ \\
   & $[[10,6,3/2]]_{4}$	      & $[[20,6,6/2]]_{4}$  &    & $[[17,12,4/2]]_{4}$	      & $[[34,12,8/2]]_{4}$ \\
$11$ & $[[11,2,7/2]]_{4}$     & $[[22,2,14/2]]_{4}$ &    & $[[17,13,3/2]]_{4}$	      & $[[34,13,6/2]]_{4}$ \\
   & $[[11,4,6/2]]_{4}$	      & $[[22,4,12/2]]_{4}$ & $18$ & $[[18,5,10/2]]_{4}$	      & $[[36,5,20/2]]_{4}$ \\
   & $[[11,5,5/2]]_{4}$	      & $[[22,5,10/2]]_{4}$ &    & $[[18,6,9/2]]_{4}$             & $[[36,6,18/2]]_{4}$ \\

   & $[[11,6,4/2]]_{4}$	      & $[[22,6,8/2]]_{4}$  &    & $[[18,8,8/2]]_{4}$         & $[[36,8,16/2]]_{4}$ \\
   & $[[11,7,3/2]]_{4}$	      & $[[22,7,6/2]]_{4}$  &    & $[[18,10,6/2]]_{4}$         & $[[36,10,12/2]]_{4}$ \\
$12$ & $[[12,2,8/2]]_{4}$     & $[[24,2,16/2]]_{4}$ &    & $[[18,11,5/2]]_{4}$         & $[[36,11,10/2]]_{4}$ \\
   & $[[12,3,7/2]]_{4}$	      & $[[24,3,14/2]]_{4}$ &    & $[[18,12,4/2]]_{4}$         & $[[36,12,8/2]]_{4}$ \\
   & $[[12,5,6/2]]_{4}$	      & $[[24,5,12/2]]_{4}$ &    & $[[18,14,3/2]]_{4}$         & $[[36,14,6/2]]_{4}$ \\

   & $[[12,7,4/2]]_{4}$	      & $[[24,7,8/2]]_{4}$  & $19$ & $[[19,4,11/2]]_{4}$         & $[[38,4,22/2]]_{4}$ \\
   & $[[12,8,3/2]]_{4}$	      & $[[24,8,6/2]]_{4}$  &    & $[[19,5,10/2]]_{4}$         & $[[38,5,20/2]]_{4}$ \\
$13$ & $[[13,2,9/2]]_{4}$     & $[[26,2,18/2]]_{4}$ &    & $[[19,6,9/2]]_{4}$         & $[[38,6,18/2]]_{4}$ \\
   & $[[13,4,7/2]]_{4}$	      & $[[26,4,14/2]]_{4}$ &    & $[[19,8,8/2]]_{4}$         & $[[38,8,16/2]]_{4}$ \\
   & $[[13,5,6/2]]_{4}$	      & $[[26,5,12/2]]_{4}$ &    & $[[19,9,7/2]]_{4}$         & $[[38,9,14/2]]_{4}$ \\

   & $[[13,6,5/2]]_{4}$	      & $[[26,6,10/2]]_{4}$ &    & $[[19,11,6/2]]_{4}$         & $[[38,11,12/2]]_{4}$ \\
   & $[[13,8,4/2]]_{4}$	      & $[[26,8,8/2]]_{4}$  &    & $[[19,12,5/2]]_{4}$         & $[[38,12,10/2]]_{4}$ \\
   & $[[13,9,3/2]]_{4}$	      & $[[26,9,6/2]]_{4}$  &    & $[[19,13,4/2]]_{4}$         & $[[38,13,8/2]]_{4}$ \\
$14$ & $[[14,2,10/2]]_{4}$    & $[[28,2,20/2]]_{4}$ &    & $[[19,15,3/2]]_{4}$         & $[[38,15,6/2]]_{4}$ \\
   & $[[14,3,9/2]]_{4}$	      & $[[28,3,18/2]]_{4}$ & $20$ & $[[20,4,12/2]]_{4}$         & $[[40,4,24/2]]_{4}$ \\

   & $[[14,4,8/2]]_{4}$	      & $[[28,4,16/2]]_{4}$ &    & $[[20,5,11/2]]_{4}$         & $[[40,5,22/2]]_{4}$ \\
   & $[[14,5,7/2]]_{4}$	      & $[[28,5,14/2]]_{4}$ &    & $[[20,6,10/2]]_{4}$         & $[[40,6,20/2]]_{4}$ \\
   & $[[14,6,6/2]]_{4}$	      & $[[28,6,12/2]]_{4}$ &    & $[[20,7,9/2]]_{4}$         & $[[40,7,18/2]]_{4}$ \\
   & $[[14,7,5/2]]_{4}$	      & $[[28,7,10/2]]_{4}$ &    & $[[20,9,8/2]]_{4}$         & $[[40,9,16/2]]_{4}$ \\
   & $[[14,9,4/2]]_{4}$	      & $[[28,9,8/2]]_{4}$  &    & $[[20,10,7/2]]_{4}$         & $[[40,10,14/2]]_{4}$ \\

   & $[[14,10,3/2]]_{4}$      & $[[28,10,6/2]]_{4}$ &    & $[[20,12,6/2]]_{4}$         & $[[40,12,12/2]]_{4}$ \\
$15$ & $[[15,2,11/2]]_{4}$    & $[[30,2,22/2]]_{4}$ &    & $[[20,13,5/2]]_{4}$         & $[[40,13,10/2]]_{4}$ \\
   & $[[15,3,10/2]]_{4}$      &	$[[30,3,20/2]]_{4}$ &    & $[[20,14,4/2]]_{4}$         & $[[40,14,8/2]]_{4}$ \\
   & $[[15,6,7/2]]_{4}$	      &	$[[30,6,14/2]]_{4}$ &    & $[[20,16,3/2]]_{4}$         & $[[40,16,6/2]]_{4}$ \\
%here is the count to 44
\hline\hline
\end{tabular}
\end{table}

\begin{remark}
Aside from its nice structural property, 
the advantage of using the mapping $S$ can be seen, 
for instance, from the fact that we have the $[[18,2,12/2]]_{4}$-code 
$Q_{S}$ which cannot be derived directly from the best-known 
linear codes for $n=18$. Similarly for the following $Q_{S}$ codes: 
$[[30,2,22/2]]_{4}$, $[[30,3,20/2]]_{4}$, $[[32,3,22/2]]_{4}$, 
$[[38,4,22/2]]_{4}$, $[[40,4,24/2]]_{4}$, $[[42,4,26/2]]_{4}$, 
$[[44,4,28/2]]_{4}$, and $[[46,4,28/2]]_{4}$.
\end{remark}

\subsection{Construction from concatenated Reed-Solomon codes}
\label{subsec:RSCodes}
Let $m$ be a positive integer. Concatenation is used to obtain 
codes over $\F_{q}$ from codes over an extension $\F_{q^{m}}$ 
of $\F_{q}$. A general method of performing concatenation is 
presented in~\cite[Sec. 6.3]{LX04} and in~\cite[Ch. 10]{MS77}.

Our strategy here is to construct nested codes $C \subset D$ 
over $\F_{4}$ from nested codes $A \subset B$ over $\F_{4^{m}}$. 
We then use the codes $C$ and $D$ and the mapping $S$ to get 
a quantum code $Q$.

The field $\F_{4^{m}}$ can be 
viewed as an $\F_{4}$-vector space with basis 
$\left\lbrace \beta_{1} \ldots,\beta_{m}\right\rbrace $. 
Then, an element $x \in \F_{4^{m}}$ can be written uniquely as

\begin{equation*}
x = \sum_{j=1}^{m} a_{j} \beta_{j} \text{ with } a_{j} \in \F_{4} \text{.}
\end{equation*}

We define a mapping 
$\phi : \F_{4^{m}}\rightarrow\F_{4}^{m}$ given by 
$x \mapsto (a_{1},\ldots,a_{m})$. This mapping is a bijective 
$\F_{4}$-linear transformation and extends naturally to 
the mapping $\phi^{*}$

\begin{equation}\label{eq:5.9}
\begin{aligned}
\phi^{*} :\quad &\F_{4^{m}}^{N}\rightarrow\F_{4}^{mN}\\
&(x_{1},\ldots,x_{n})\mapsto(\phi(x_{1}),\ldots,\phi(x_{n})) \text{.}
\end{aligned}
\end{equation}

If $A$ is an $[N,K,D]_{4^{m}}$-code and letting $C = \phi^{*}(A)$, 
then it is easy to verify that $C$ is an $[mN,mK,\geq D]_{4}$-code. 
Moreover, the mapping $\phi^{*}$ preserves nestedness by its 
$\F_{4}$-linearity. That is, if an $[N,K_{1},D_{1}]_{4^{m}}$-code $A$ 
is a subcode of an $[N,K_{2},D_{2}]_{4^{m}}$-code $B$, 
then $C=\phi^{*}(A)$ is a subcode of $D=\phi^{*}(B)$ as codes over $\F_{4}$.

Let $q=4^{m}$ and $\alpha_{1},\ldots,\alpha_{q-1}$ be the 
nonzero elements of $\F_{q}$. It is well-known (see, 
e.g.,~\cite[Ch. 10 and Ch.11]{MS77}) that the $[q,k,q-k+1]_{q}$-extended 
Reed-Solomon (henceforth, RS) code $B$ has a parity check matrix
\begin{equation}\label{eq:5.10}
H=\left(
\begin{array}{*{12}{l}}
1                   & 1                  & \ldots   & 1                    & 1 \\
\alpha_{1}          & \alpha_{2}         & \ldots   & \alpha_{q-1}         & 0 \\
\alpha_{1}^{2}      & \alpha_{2}^{2}     & \ldots   & \alpha_{q-1}^{2}     & 0 \\
\vdots	            & \vdots             & \ddots   & \vdots               & \vdots \\
\alpha_{1}^{q-k-1}  & \alpha_{2}^{q-k-1} & \ldots   & \alpha_{q-1}^{q-k-1} & 0 
\end{array}
\right).
\end{equation}

Let $A$ be the $[q,1,q]_{q}$-repetition code generated by 
$\textbf{1}=(1,\ldots,1)$. For $1 \leq j \leq q-2$, the sum 
$s=\sum_{l=1}^{q-1}\alpha_{l}^{j}=0$. To see this, we choose 
$\alpha$ a primitive element of $\F_{q}$. Then 
$\alpha^{j}s=\sum_{l=1}^{q-1}(\alpha \cdot \alpha_{l})^{j}=s$. 
Since $\alpha^{j} \neq 1$, we conclude that $s=0$. 
This implies that $A \subset B$.

Note that we can choose an $\F_{4}$-basis 
$\left\lbrace \beta_{1} \ldots,\beta_{m}\right\rbrace$ of 
$\F_{q}$ such that a generator matrix of $C'=\phi^{*}(A)$ is 
given by the $m \times mq$ matrix 
$G=\left( I_{m}|I_{m}|\ldots|I_{m}\right )$ where $I_{m}$ is 
the $m \times m$ identity matrix. Hence, $C'$ is of parameters 
$[mq,m,q]_4$. Define $C$ to be the $[mq,1,mq]_4$-repetition 
code subset of $C'$. This is valid since we know that 
$\textbf{1}=(1,\ldots,1)\in C'$. The code $D=\phi^{*}(B)$ 
is an $[mq,mk,d' \geq (q-k+1)]_{4}$-code that contains $C$. 
Repeating the proof of Theorem~\ref{theorem:6.1} 
yields the following result.

\begin{theorem}\label{theorem:7.3}
Let $m$ be a positive integer, $q=4^{m}$, and $1 \leq k \leq q$. 
Then there exists a $[[2mq,mk-1,(\geq 2(q-k+1))/2]]_{4}$-asymmetric 
quantum code $Q$.
\end{theorem}

\begin{remark}\label{remark:7.4}
For a specific value of $m$ and a given basis 
$\left\lbrace \beta_{1},\ldots,\beta_{m}\right\rbrace $, 
$d'=d(D)$ can be explicitly computed. As noted in~\cite[Ch. 10]{MS77}, 
a change of basis may change the weight distribution and even 
the minimum weight of the code $D$.
\end{remark}

\begin{example}\label{example:7.5}
For $m=2$ and $1 \leq k \leq 16$ we get the 
$[[64,k',d_{z}/2]]_{4}$-quantum codes listed in Table~\ref{table:RS16}.

\begin{table}[h]
\caption{$[[64,k',d_{z}/2]]_{4}$-code $Q$ from $[16,k,16-k+1]_{16}$-extended RS codes}
\label{table:RS16}
\centering
\begin{tabular}{| c | l | l | l | l | l | l | l | l | l |} % (9 columns)
\hline\hline
$k$          & $1$  & $2$  & $3$  & $4$  & $5$  & $6$  & $7$  & $8$\\
$k'$         & $1$  & $3$  & $5$  & $7$  & $9$  & $11$ & $13$ & $15$\\
$d_{z} \geq$ & $32$ & $30$ & $28$ & $26$ & $24$ & $22$ & $20$ & $18$\\
\hline\hline
$k$          & $9$   & $10$  & $11$  & $12$ & $13$ & $14$ & $15$ & $16$\\
$k'$         & $17$  & $19$  & $21$  & $23$ & $25$ & $27$ & $29$ & $31$\\
$d_{z} \geq$ & $16$  & $14$  & $12$  & $10$ & $8$  & $6$  & $4$  & $2$\\
\hline\hline
\end{tabular}
\end{table}
\end{example}

\begin{example}\label{example:6.6}
For $m=3$ and $1 \leq k \leq 64$ we get the 
$[[384,k',d_{z}/2]]_{4}$-quantum codes listed in Table~\ref{table:RS64}.

\begin{table}[h]
\caption{$[[384,k',d_{z}/2]]_{4}$-code $Q$ from $[64,k,64-k+1]_{64}$-extended RS codes}
\label{table:RS64}
\centering
\begin{tabular}{| c | l | l | l | l | l | l | l | l | l |} % (9 columns)
\hline\hline %inserts double horizontal lines
$k$          & $1$   & $2$  & $3$  & $4$  & $5$  & $6$  & $7$  & $8$\\
$k'$         & $2$   & $5$  & $8$  & $11$  & $14$ & $17$ & $20$ & $23$\\
$d_{z} \geq$ & $128$ & $126$  & $124$  & $122$   & $120$  & $118$  & $116$  & $114$\\
\hline\hline % inserts double horizontal line
$k$          & $9$   & $10$  & $11$  & $12$ & $13$ & $14$ & $15$ & $16$\\
$k'$         & $26$  & $29$  & $32$  & $35$ & $38$ & $41$ & $44$ & $47$\\
$d_{z} \geq$ & $112$   & $110$   & $108$   & $106$  & $104$  & $102$ & $100$ & $98$\\
\hline\hline %inserts double horizontal lines
$k$          & $17$  & $18$  & $19$  & $20$  & $21$  & $22$ & $23$ & $24$\\
$k'$         & $50$  & $53$  & $56$  & $59$  & $62$  & $65$ & $68$ & $71$\\
$d_{z} \geq$ & $96$    & $94$    & $92$    & $90$    & $88$    & $86$   & $84$   & $82$\\
\hline\hline % inserts double horizontal line
$k$          & $25$  & $26$  & $27$  & $28$ & $29$ & $30$ & $31$ & $32$\\
$k'$         & $74$  & $77$  & $80$  & $83$ & $86$ & $89$ & $92$ & $95$\\ % inserting body of the table
$d_{z} \geq$ & $80$    & $78$    & $76$    & $74$   & $72$   & $70$   & $68$   & $66$\\
\hline\hline %inserts double horizontal lines
$k$          & $33$    & $34$     & $35$     & $36$     & $37$     & $38$    & $39$    & $40$\\
$k'$         & $98$  & $101$  & $104$  & $107$  & $110$  & $113$ & $116$ & $119$\\ % inserting body of the table
$d_{z} \geq$ & $64$    & $62$     & $60$     & $58$     & $56$     & $54$    & $52$    & $50$\\
\hline\hline %inserts double horizontal lines
$k$          & $41$     & $42$     & $43$     & $44$     & $45$     & $46$    & $47$    & $48$\\
$k'$         & $122$  & $125$  & $128$  & $131$  & $134$  & $137$ & $140$ & $143$\\ % inserting body of the table
$d_{z} \geq$ & $48$     & $46$     & $44$     & $42$     & $40$     & $38$    & $36$    & $34$\\
\hline\hline %inserts double horizontal lines
$k$          & $49$     & $50$     & $51$     & $52$     & $53$     & $54$    & $55$    & $56$\\
$k'$         & $146$  & $149$  & $152$  & $155$  & $158$  & $161$ & $164$ & $167$\\ % inserting body of the table
$d_{z} \geq$ & $32$     & $30$     & $28$     & $26$     & $24$     & $22$    & $20$    & $18$\\
\hline\hline % inserts double horizontal line
$k$          & $57$   & $58$  & $59$  & $60$ & $61$ & $62$ & $63$ & $64$\\
$k'$         & $170$  & $173$  & $176$  & $179$ & $182$ & $185$ & $188$ & $191$\\ % inserting body of the table
$d_{z} \geq$ & $16$  & $14$  & $12$  & $10$ & $8$  & $6$  & $4$  & $2$\\
\hline\hline %inserts double line
\end{tabular}
\end{table}
\end{example}

\section{Conclusions and Open Problems}\label{sec:Conclusion}
In this paper we have given a special construction of asymmetric quantum
codes. An analysis on the weight enumerators of the resulting quantum 
codes is also presented. It seems that the construction is especially 
useful when the constraint on $d_{x}$ is minimal and the demand on 
$d_{z}$ is critical.

This allows us to give a more general criterion to use in choosing 
a pair of $\F_{4}$-linear codes $C \subset D$ that, in some cases, 
yields asymmetric quantum codes with improved parameters compared to 
those listed in~\cite{ELS10}. Many new asymmetric quantum codes are 
also found.

There are direct generalizations of the mapping $S$. One direction might be to
use non-quadratic extensions. Another one is to generalize it to fields of 
odd characteristics. The latter might be more promising than the former.

\section*{Acknowledgment}\label{sec:Acknowledge}
The work of M.~F.~Ezerman was carried out under the Nanyang Technological 
University PhD Research Scholarship. The work of S.~Ling and P.~Sol{\'e} was 
partially supported by Singapore National Research Foundation Competitive 
Research Program grant NRF-CRP2-2007-03 and by the Merlion Programme 01.01.06. 
P.~Sol{\'e} acknowledges the hospitality of the Department of Mathematics at 
El Manar Tunis where part of the research was done. Likewise, O.~Yemen 
is grateful for the hospitality she experienced at the I3S-CNRS Laboratory 
at Sophia Antipolis. Her work was supported by the Algebra and Number Theory 
Laboratory 99/UR/15-18, the Faculty of Sciences of Tunis.

% You may incorporate your references as follows in your main tex file.
% Using BibTex is not recommended but can be handled.

\medskip
% The data information below will be filled by AIMS editorial staff
Received xxxx 20xx; revised xxxx 20xx.
\medskip

\end{document}